\documentclass[12pt]{amsart}
\usepackage{}

\usepackage{amsmath}
\usepackage{amsfonts}
\usepackage{amssymb}
\usepackage[all]{xy}           
\usepackage{bbm}
\usepackage{bbding}
\usepackage{txfonts}
\usepackage{amscd}
\usepackage{tikz}
\usetikzlibrary{matrix}
\usepackage[shortlabels]{enumitem}

\usepackage{ifpdf}
\ifpdf
\usepackage[colorlinks,final,backref=page,hyperindex]{hyperref}
\else
\usepackage[colorlinks,final,backref=page,hyperindex,hypertex]{hyperref}
\fi
\usepackage{tikz}
\usepackage[active]{srcltx}

\newcommand{\rmnum}[1]{\romannumeral #1}
\newcommand{\Rmnum}[1]{\expandafter\@slowromancap\romannumeral #1@}

\makeatletter

\newtheorem{thm}{Theorem}[section]
\newtheorem{lem}[thm]{Lemma}

\newtheorem{pro}[thm]{Proposition}
\newtheorem{ex}[thm]{Example}
\newtheorem{rmk}[thm]{Remark}
\newtheorem{defi}[thm]{Definition}

\newcommand {\emptycomment}[1]{}

\newcommand{\be }{\begin{equation}}
\newcommand{\ee }{\end{equation}}

\newcommand{\br}[1]{   [ \cdot,    \cdot  ]   }

\newcommand\blfootnote[1]{%
  \begingroup
  \renewcommand\thefootnote{}\footnote{#1}%
  \addtocounter{footnote}{-1}%
  \endgroup
}

\begin{document}

\setlength{\baselineskip}{1.2\baselineskip}
\title[Universal enveloping algebras of weighted differential Poisson algebras]
{Universal enveloping algebras of weighted differential Poisson algebras}

\author{Ying Chen}
\address{
School of Mathematical Sciences  \\
Zhejiang Normal University\\
Jinhua 321004\\
China}
\email{yingchen@zjnu.edu.cn}

\author{Chuangchuang Kang}
\address{
School of Mathematical Sciences  \\
Zhejiang Normal University\\
Jinhua 321004 \\
China}
\email{kangcc@zjnu.edu.cn}

\author{Jiafeng L\"u}
\address{
School of Mathematical Sciences    \\
Zhejiang Normal University\\
Jinhua 321004              \\
China}
\email{jiafenglv@zjnu.edu.cn}
\blfootnote{*Corresponding Author: Chuangchuang Kang. Email: kangcc@zjnu.edu.cn.}
\begin{abstract}
The $\lambda$-differential operators and modified $\lambda$-differential operators are generalizations of classical differential operators.
This paper introduces the notions of $\lambda$-differential Poisson ($\lambda$-DP for short) algebras and modified $\lambda$-differential Poisson ($\lambda$-mDP for short) algebras as generalizations of differential Poisson algebras. The $\lambda$-DP algebra is proved to be closed under tensor product, and a $\lambda$-DP  algebra structure is provided on the cohomology algebra of the $\lambda$-DP algebra.
These conclusions are also applied to $\lambda$-mDP algebras and their modules. Finally, the universal enveloping algebras of $\lambda$-DP algebras are generalized by constructing a $\mathcal{P}$-triple. Three isomorphisms among opposite algebras, tensor algebras and the universal enveloping algebras of $\lambda$-DP algebras are obtained.
\end{abstract}



\maketitle
\tableofcontents
\allowdisplaybreaks

\section{Introduction and Main Statements}\label{sec:intr}

This paper aims to study $\lambda$-differential Poisson algebras, modified $\lambda$-differential Poisson algebras, and their universal enveloping algebras.
\subsection{Differential algebras}
The notion of differential algebras was first introduced by J. F. Ritt \cite{Ritt. differ-alg} in 1950, which originated from the algebraic study of polynomial ordinary and partial differential equations. 
\begin{defi}
A \textbf{differential algebra} is an algebra $(A,\cdot)$ endowed with a differential operator $d:A\rightarrow A$ which satisfies the Leibniz rule
$$d(x\cdot y)=d(x)\cdot y+x\cdot d(y),\quad\forall~x,y\in A.$$  
\end{defi}
For any commutative algebra $A$, the polynomial ring $A[x]$ is a differential algebra with a differential operator $d$ defined as the standard differentiation of polynomials:
$$
d(\sum_{k=1}^na_kx^k)=\sum_{k=1}^na_kkx^{k-1}.
$$
In differential geometry, for the commutative algebra $C^{\mathcal{\infty}}(M)$ consisting of the smooth functions on a smooth real manifold $M$, the differential operators of $C^{\mathcal{\infty}}(M)$ are in bijective correspondence with smooth vector fields of $M$ \cite{Lee}. 
For a module $M$, there exists a differential algebra $\mathrm{Diff}(M)$ which is the free differential algebra over $M$ \cite{Kolchin}. Differential algebras also appear in Galois theory \cite{Galois}, Hurwitz series algebra \cite{Keigher}, operad theory \cite{operad},  De Rham cohomology \cite{Weibel},  differential categories\cite{Blute}, codifferential categories \cite{category} and etc. For further details on differential algebras, see \cite{Kaplansky,Kolchin}.

\subsection{$\lambda$-differential algebras and $\lambda$-differential Lie algebras}
The important generalizations of differential algebras are  weighted differential algebras. In 2008, L. Guo and W. Keigher \cite{Guo. dRB} defined differential operators with weight $\lambda\in \mathbf{k}$ as an analogue of $\lambda$-Rota Baxter operators. 
\begin{defi}
A \textbf{$\lambda$-differential algebra} is a triple $(A,\cdot,d)$, where $(A,\cdot)$ is an algebra together with a linear operator $d:A\rightarrow A$ satisfying $d(1_A) = 0$ and the weighted Leibniz rule
\begin{eqnarray}
  \label{eq:lambda} d(x\cdot y) &=& d(x)\cdot y+x\cdot d(y)+\lambda d(x)\cdot d(y),\quad\forall x,y\in A.
\end{eqnarray}
\end{defi}
Particularly, $0$-differential algebras are the classical  differential algebras and $1$-differential algebras are the difference algebras \cite{difference}.
Then in \cite{Guo. integro-diff-alg}, the authors generalized weighted integro-differential algebras and constructed free commutative case generated by a differential algebra. 

In the Lie theory, a $\lambda$-differential Lie algebra \cite{Unen-dLie} was obtained by replacing the multiplication in the $\lambda$-differential algebra $(A,\cdot,d)$ with a Lie bracket, which is a generalization of a LieDer pair introduced in \cite{R. Tang}. 
\begin{defi}
A \textbf{$\lambda$-differential Lie  algebra} ($\lambda$-DL algebra for short) is a triple $(A,\{\cdot,\cdot\},d)$, where $(A,\{\cdot,\cdot\})$ is a Lie algebra together with a linear map $d: A\rightarrow A$ such that
\begin{equation}\label{eq:diff-Liealg}
  d\{x,y\}=\{d(x),y\}+\{x,d(y)\}+\lambda \{d(x),d(y)\},\quad\forall x, y \in A.
\end{equation}
\end{defi}
The $0$-differential Lie algebra is the LieDer pair, and the $1$-differential Lie algebra is the difference Lie algebra that can be integrated to a relative difference Lie group in a functorial way \cite{Jiang}. In \cite{Unen-dLie}, the authors also studied the universal enveloping differential algebra of a $\lambda$-DL algebra, and  obtained the corresponding Poincaré-Birkhoff-Witt theorem. In \cite{Guo. dVirasoro}, the authors determined the $\lambda$-differential operators and $\lambda$-differential modules for the Virasoro algebras. 
In \cite{Giu. Chen.}, the authors established the Gr\"{o}bner-Shirshov bases theory for Lie $\lambda$-differential  $\Omega$-algebras. 

\subsection{Modified $\lambda$-Rota-Baxter algebras and modified $\lambda$-differential algebras}

Rota-Baxter algebras with weight $\lambda$ originated from the 1960's work \cite{Baxter} of
G. Baxter in fluctuation theory of probability, which have found remarkable
applications in various areas in mathematics and mathematical physics \cite{renormalization,shuffle-pro,dendriform-alg}.  See \cite{Guo-RB.} and the references
therein.

A $\lambda$-Rota-Baxter algebra  is an algebra $R$ with a linear map $P:R\rightarrow R$ such that 
\begin{equation*}\label{eq:RBE}
  P(x)P(y)=P(P(x)y)+P(xP(y))+\lambda P(xy),\quad\forall x, y \in R.
\end{equation*}
The Rota-Baxter operator $P$ could induce a linear map $P':=-2P-\lambda id$, which satisfies the modified Rota-Baxter equation:
\begin{equation}\label{eq:mRBE}
  P'(x)P'(y)=P'(P'(x)y)+P'(xP'(y))-\lambda^2xy,\quad\forall x, y \in R.
\end{equation}
Then the algebra $R$ together with $P'$ is a modified $(-\lambda^2)$-Rota-Baxter algebra, and $P'$ is called a modified $(-\lambda^2)$-Rota-Baxter operator. In \cite{Guo. Zhang.}, the authors constructed free commutative modified $\lambda$-Rota-Baxter algebras by the shuffle product and applied a Hochschild cocycle to show these modified $\lambda$-Rota-Baxter algebras contain a Hopf algebra structure. There are further studies on modified $\lambda$-Rota-Baxter algebras in \cite{Das,Zhu}.

Back in the 1980s, M. A. Semenov-Tyan-Shanski\u{\i} \cite{Semenov-TS} discovered the modified classical Yang-Baxter equation (MCYBE) on a Lie algebra $(A,\{\cdot,\cdot\})$,
\begin{equation*}\label{eq:mCYBE}
  \{P(x),P(y)\}=P\{P(x),y\}+P\{x,P(y)\}-\{x,y\}, \quad\forall x, y \in A,
\end{equation*}
whose solutions are called modified $r$-matrices. If the algebra in \eqref{eq:mRBE} is a Lie algebra, then the MCYBE is the special case of \eqref{eq:mRBE} by setting $\lambda=1$. It is also closely related to the operator forms of classical Yang-Baxter equation \cite{Bai-C} and generalized Lax pairs \cite{Bordemann M.}.  

These ideas were also used to construct modified $\lambda$-differential algebras. For any differential algebra $(A,\cdot,d)$, a linear map $d':=d-\lambda id$ satisfies the following identity 
\begin{equation}\label{eq:modi-lambda}
  d'(x\cdot y) = d'(x)\cdot y+x\cdot d'(y)+\lambda x\cdot y,\quad\forall x, y \in A.
\end{equation}
If the map of the last term on the right side of \eqref{eq:lambda} is an identity map, then \eqref{eq:lambda} equals \eqref{eq:modi-lambda}. 

\begin{defi}
A linear map $d':A\rightarrow A$ satisfying \eqref{eq:modi-lambda} is called a modified $\lambda$-differential operator, and $(A,\cdot,d')$ is a \textbf{modified $\lambda$-differential algebra} ($\lambda$-mD algebra for short) .
\end{defi}
The $\lambda$-mD algebra is used to construct weighted infinitesimal bialgebras \cite{infi-bialgebra, infi-Hopf}.
In \cite{Unen-dLie}, the authors replaced the algebra of \eqref{eq:modi-lambda} with a Lie algebra and studied the \textbf{modified $\lambda$-differential Lie algebras} ($\lambda$-mDL algebra for short). They also studied the Wronskian envelope of a $\lambda$-mDL algebra and demonstrated that these results do not hold for  the $\lambda$-DL algebras.

\subsection{Universal enveloping algebras of Poisson algebras}
The Poisson algebras originated from the study of Hamiltonian mechanics by S. D. Poisson,
which play a crucial role in broad areas including Poisson manifolds \cite{Poi-manifold}, quantum groups \cite{quantum-G}, symplectic geometry \cite{symplectic}, quantization theory \cite{quantization} and classical quantum mechanics \cite{quantum-mecha}.

S. Q. Oh \cite{Oh. Poi-enve-alg} first proposed the universal enveloping algebra of a Poisson algebra by a universal property of morphisms, which is analogous to the universal enveloping algebra of a Lie algebra \cite{Milnor. Lie-enve-alg}.  The category of Poisson modules  has been proven to be isomorphic to the category of  modules over universal enveloping algebra of a Poisson algebra.
U. Umirbaev \cite{Umirbaev} explicitly constructed  the universal enveloping algebras of Poisson algebras by generators and defining relations. Moreover, in \cite{Lv. Poi-Ore-ex},
the authors constructed universal enveloping algebras of Poisson-Ore extensions. They also showed that the universal enveloping algebra of a Poisson Hopf algebra is isomorphic to the crossed product of a pointed Poisson Hopf algebra and a quotient Hopf algebra \cite{Lv. Poi-Hopf-alg}.  

\subsection{Motivations}
In \cite{Lv}, the authors introduced the notions of differential graded Poisson algebras. They studied the universal property of universal enveloping algebras of differential graded Poisson algebras, and proved that there is an isomorphism between the category of differential graded Poisson modules and the category of differential graded modules. The examples of differential graded Poisson algebras have also been found in Lie theory, differential geometry, homological algebra, and deformation theory \cite{Beltran,Lv}. The PBW-basis for the universal enveloping algebras of differential graded Poisson algebras is studied in \cite{Xianguo Hu}. The differential graded Poisson algebras were further studied in \cite{XiaojieLi,Caradot}.

\begin{defi}\cite{Lv}
Let $(A, \cdot)$ be a graded $\mathbf{k}$-algebra. A \textbf{graded Poisson algebra} is a triple $(A, \cdot, [\cdot,\cdot])$, where $[\cdot, \cdot]: A\otimes A\rightarrow A$ is a $\mathbf{k}$-linear map of degree $p$, for all homogeneous elements $a, b, c \in A$,  such that 
\begin{enumerate}
\item[(\rmnum{1})] $(A, [\cdot, \cdot])$ is a graded Lie algebra, i.e. the bracket $[\cdot, \cdot]$ satisfies 
\begin{enumerate}
\item[] $[a, b]=-(-1)^{(|a|+p)(|b|+p)}[b, a]$,

\item[] $[a, [b, c]]=[[a, b], c]+(-1)^{(|a|+p)(|b|+p)}[b, [a, c]]$;
\end{enumerate}

\item[(\rmnum{2})] (graded commutativity): $a\cdot b=(-1)^{|a||b|}b\cdot a$;

\item[(\rmnum{3})] (biderivation property): $[a, b\cdot c]=[a, b]\cdot c+(-1)^{(|a|+p)|b|}b\cdot [a, c]$.
\end{enumerate}
In addition, if there is a $\mathbf{k}$-linear homogeneous map $d: A\rightarrow A$ of degree 1 such
that $d^2=0$ and
\begin{enumerate}
\item[(\rmnum{4})] $d([a, b])=[d(a), b]+(-1)^{(|a|+p)}[a, d(b)]$;

\item[(\rmnum{5})] $d(a\cdot b)=d(a)\cdot b+(-1)^{|a|}a\cdot d(b)$,
\end{enumerate}
 then $(A, \cdot, [\cdot,\cdot], d)$ is called a \textbf{differential graded Poisson algebra}.
\end{defi}

A differential Poisson algebra can be considered as a differential graded Poisson algebra on the non-graded vector spaces. Note that the operators in the differential Poisson algebra are both differential operators of $(A, [\cdot, \cdot])$ and $(A,\cdot)$ with weight $0$.
For any arbitrary nonzero weight $\lambda\in \mathbf{k}$, it is natural to investigate the weighted differential Poisson structure and its universal enveloping algebras, which include $\lambda$-differential operators and modified  $\lambda$-differential operators.


\subsection{Outline of this paper}
The paper is organized as follows. 
In Section \ref{sec:DP-algs}, we introduce the notions of $\lambda$-differential Poisson ($\lambda$-DP for short) algebras and modified $\lambda$-differential Poisson ($\lambda$-mDP for short) algebras. We show that the tensor product of any two $\lambda$-DP (or $\lambda$-mDP) algebras can form a $\lambda$-DP (or $\lambda$-mDP) algebra (Proposition \ref{pro:Poi-tensor-alg}). Additionally, we show that the cohomology algebra of any $\lambda$-DP (or $\lambda$-mDP) algebra can also form a $\lambda$-DP (or $\lambda$-mDP) algebra (Propsition \ref{pro:coho-alg}).
Furthermore, we also extend these conclusions to the modules of $\lambda$-DP algebras and $\lambda$-mDP algebras. 

In Section \ref{sec:Uea-DPa}, first we introduce the notions of universal enveloping algebras of $\lambda$-DP (or $\lambda$-mDP) algebras.  It is proved that the universal enveloping algebra of any $\lambda$-DP (or $\lambda$-mDP) algebra is a (modified) $\lambda$-differential algebra (Proposition \ref{pro:uni-enve-alg}). 
Then we introduce a $\mathcal{P}$-triple of a $\lambda$-DP (or $\lambda$-mDP) algebra, 
and prove that there is a unique (modified) $\lambda$-differential algebra homomorphism between a universal enveloping algebra of a $\lambda$-DP (or $\lambda$-mDP) algebra and a $\mathcal{P}$-triple of a $\lambda$-DP (or $\lambda$-mDP) algebra (Theorem \ref{thm:bi-commute}).
For any two $\lambda$-DP (or $\lambda$-mDP) algebras and their respective corresponding $\mathcal{P}$-triples, we construct a $\mathcal{P}$-triple of tensor algebra of these two $\lambda$-DP (or $\lambda$-mDP) algebras. Furthermore, three isomorphisms among opposite algebras of $\lambda$-DP (or $\lambda$-mDP) algebras, tensor algebras of $\lambda$-DP (or $\lambda$-mDP) algebras and universal enveloping algebras of $\lambda$-DP (or $\lambda$-mDP) algebras are obtained in Theorem \ref{thm:oppo-tensorpro}.

{\bf Notation} Throughout this paper, let $\mathbf{k}$ be a field of characteristic $0$. Except specially stated, $\lambda$ is fixed in $\mathbf{k}$.

\section{$\lambda$-DP algebras, $\lambda$-mDP algebras and their modules}\label{sec:DP-algs}
In this section, we study $\lambda$-DP algebras,~$\lambda$-mDP algebras and their modules. 


\subsection{$\lambda$-DP algebras and $\lambda$-mDP algebras}\label{subsec-DP alg}
\begin{defi}\label{defi:modi-diff-Poialg}
Let $(A,\cdot,d)$ be a $\lambda$-differential algebra, and $(A,\{\cdot,\cdot\},d)$ be a $\lambda$-DL algebra. 
A {\bf $\lambda$-DP algebra} is a quadruple $(A,\cdot,\{\cdot,\cdot\},d)$ satisfying:
\begin{equation}\label{eq:Leibniz-rule}
  \{a,b\cdot c\}=\{a,b\}\cdot c+b\cdot\{a,c\},\quad\forall a,b,c \in A.
\end{equation}
\end{defi}
If $(A,\cdot,d)$ is a modified $\lambda$-differential algebra and $(A,\{\cdot,\cdot\},d)$ is a $\lambda$-mDL algebra, then the quadruple $(A,\cdot,\{\cdot,\cdot\},d)$ satisfying \eqref{eq:Leibniz-rule} is a {\bf $\lambda$-mDP algebra}. 

\begin{ex}
For any $\lambda$-differential algebra $(A,\cdot,d)$, if a bilinear map $\{\cdot, \cdot\}:A\otimes A\rightarrow A$ is defined by 
$$\{x,y\}:=x\cdot y-y\cdot x,\quad \forall x,y \in A,$$
then $(A,\cdot,\{\cdot, \cdot\},d)$ is a $\lambda$-DP algebra. If replacing $d$ with a modified $\lambda$-differential operator, then $(A,\cdot,\{\cdot, \cdot\},d)$ is a $\lambda$-mDP algebra.
\end{ex}
\begin{proof}
For all $x,y,z\in A$, we have
\begin{eqnarray*}
  d(\{x,y\}) &=& d(x\cdot y-y\cdot x)=d(x\cdot y)-d(y\cdot x) \\
  &=& d(x)\cdot y+x\cdot d(y)+\lambda d(x)\cdot d(y)-d(y)\cdot x-y\cdot d(x)-\lambda d(y)\cdot d(x)\\
  &=& \big(d(x)\cdot y-y\cdot d(x)\big)+\big(x\cdot d(y)-d(y)\cdot x\big)
  +\lambda \big(d(x)\cdot d(y)-d(y)\cdot d(x)\big)\\
  &=&\{d(x),y\}+\{x,d(y)\}+\lambda \{d(x),d(y)\},
\end{eqnarray*}
and
\begin{eqnarray*}
  \{x,y\cdot z\} &=& x\cdot(y\cdot z)-(y\cdot z)\cdot x \\
  &=& (x\cdot y)\cdot z-(y\cdot z)\cdot x+y\cdot(x\cdot z)-y\cdot(z\cdot x) \\
  &=&\{x,y\}\cdot z+y\cdot\{x,z\}.
\end{eqnarray*}
Therefore, $(A,\cdot,\{\cdot, \cdot\},d)$ is a $\lambda$-DP algebra.  The modified  case is proved similarly.
\end{proof}

\begin{ex}
 A $\lambda$-DL algebra $(sl_2(\mathbb{C}),\{\cdot,\cdot\},d_1)$, and a $\lambda$-mDL algebra $(sl_2(\mathbb{C}),\{\cdot,\cdot\},d_2)$ were given in \cite{Unen-dLie}, where $sl_2(\mathbb{C})$ is a simple Lie algebra generated linearly from elements $\{h, e, f\}$. The Lie bracket, $\lambda$-differential operator and modified $\lambda$-differential operator on $sl_2(\mathbb{C})$ are defined as follows:
\begin{eqnarray*}
  && \{h,e\}=2e, \{h,f\}=-2f, \{e,f\}=h. \\
  && d_1(h):=2e,\quad d_1(e):=0,\quad d_1(f):=-h-\lambda e. \\
  && d_2(h):=-\lambda h+2e+2f,\quad d_2(e):=-h-\lambda e,\quad d_2(f):=-h-\lambda f.
\end{eqnarray*}
Moreover, if the bilinear map $\cdot:sl_2(\mathbb{C})\times sl_2(\mathbb{C})\rightarrow sl_2(\mathbb{C})$ satisfying
\begin{eqnarray*}
  && h\cdot h=0,\quad h\cdot e=ke,\quad h\cdot f=-kf, \\
  && e\cdot h=-ke,\quad e\cdot e=0,\quad e\cdot f=\frac{k}{2}h, \\
  && f\cdot h=kf,\quad f\cdot e=-\frac{k}{2}h,\quad f\cdot f=0,\quad k \in \mathbb{C}. 
\end{eqnarray*}
Then $(sl_2(\mathbb{C}),\cdot,\{\cdot,\cdot\},d_1)$ is a $\lambda$-DP algebra, 
and $(sl_2(\mathbb{C}),\cdot,\{\cdot,\cdot\},d_2)$ is a $\lambda$-mDP algebra.
\end{ex}

(Modified) $\lambda$-differential subalgebras and (modified) $\lambda$-differential ideals were proposed in \cite{Unen-dLie}. We add the Lie algebraic structure and give the definitions of $\lambda$-DP (or $\lambda$-mDP) subalgebras and $\lambda$-DP (or $\lambda$-mDP) ideals.

\begin{defi}
Let $(A,\cdot,\{\cdot,\cdot\},d)$ be a $\lambda$-DP (or $\lambda$-mDP) algebra, and $(B,\cdot)$ is a subalgebra of $(A,\cdot)$.
\begin{enumerate}
  \item If $d(B)\subseteq B$ and $\{B,B\}\subseteq B$, then $(B,\cdot,\{\cdot,\cdot\},d)$ is called a $\lambda$-DP (or $\lambda$-mDP) subalgebra of $(A,\cdot,\{\cdot,\cdot\},d)$, denoted by $B<A$.
  \item If $d(B), \{A,B\}, \{B,A\}, A \cdot B$ and $B \cdot A$ are all contained in $B$, then $(B,\cdot,\{\cdot,\cdot\},d)$ is called a $\lambda$-DP (or $\lambda$-mDP) ideal of $(A,\cdot,\{\cdot,\cdot\},d)$, denoted by $B\lhd A$.
\end{enumerate}
\end{defi}
Obviously $0$ and $(A,\cdot,\{\cdot,\cdot\},d)$ itself are trivial $\lambda$-DP (or $\lambda$-mDP) ideals of $(A,\cdot,\{\cdot,\cdot\},d)$. 
If $(A,\cdot,\{\cdot,\cdot\},d)$ only contains trivial $\lambda$-DP (or $\lambda$-mDP) ideals, 
then we call $(A,\cdot,\{\cdot,\cdot\},d)$ a simple $\lambda$-DP (or $\lambda$-mDP) algebra.

(Modified) $\lambda$-differential algebra morphism was introduced in \cite{Unen-dLie}. If this (modified) $\lambda$-differential algebra homomorphism is also a Lie algebra homomorphism, then we call $\lambda$-DP (or $\lambda$-mDP) algebra homomorphism.

\begin{defi}
Let $(A,\cdot_A,\{\cdot,\cdot\}_A,d_A)$ and $(B,\cdot_B,\{\cdot,\cdot\}_B,d_B)$ be two $\lambda$-DP (or $\lambda$-mDP) algebras. An algebra homomorphism  $f:A\rightarrow B$ is called a $\lambda$-DP (or $\lambda$-mDP) algebra {\bf homomorphism} if $f(\{a,b\}_A)=\{f(a),f(b)\}_B$ and $f\circ d_A=d_B\circ f$ for all $a,b \in A$. Moreover, $f$ is called a $\lambda$-DP (or $\lambda$-mDP) algebra {\bf monomorphism} if $\mathrm{Ker}f=0$, a $\lambda$-DP (or $\lambda$-mDP) algebra {\bf epimorphism} if $\mathrm{Im}f=B$, a $\lambda$-DP (or $\lambda$-mDP) algebra {\bf isomorphism} if it is a bijection, denoted as $A\cong B$.
\end{defi}

\begin{pro}\label{pro:oppo-alg}
Let $(A,\cdot,\{\cdot,\cdot\},d)$ be a $\lambda$-DP (or $\lambda$-mDP) algebra, then $(A^{op},\cdot_{op},\{\cdot,\cdot\}_{op},d_{op})$ is a $\lambda$-DP (or $\lambda$-mDP) algebra, where for all $a,b \in A$,
\begin{eqnarray}
  \label{eq:asso-op}a\cdot_{op}b &:=& a\cdot b, \\
  \label{eq:diff-op}d_{op} &:=& d, \\
  \nonumber \{a,b\}_{op} &:=& -\{a,b\}. 
\end{eqnarray}

\end{pro}

\begin{proof}
By \eqref{eq:asso-op} and \eqref{eq:diff-op}, $(A^{op},\cdot_{op},d_{op})$ is a $\lambda$-differential algebra.
In fact, we have
\begin{equation*}
  \{a,b\}_{op} = -\{a,b\}=\{b,a\}=-\{b,a\}_{op},
\end{equation*}
\begin{eqnarray*}
  \{a,\{b,c\}_{op}\}_{op} &=& \{a,\{b,c\}\} \\
  &=& \{\{a,b\},c\}+\{b,\{a,c\}\} \\
  &=& \{\{a,b\}_{op},c\}_{op}+\{b,\{a,c\}_{op}\}_{op},
\end{eqnarray*}
and
\begin{eqnarray*}
  d_{op}(\{a,b\}_{op}) &=& -d(\{a,b\}) \\
  &\overset{\eqref{eq:diff-Liealg}}{=}& -\{d(a),b\}-\{a,d(b)\}-\lambda\{d(a),d(b)\} \\
  &=& \{d_{op}(a),b\}_{op}+\{a,d_{op}(b)\}_{op}+\lambda\{d(a),d(b)\}_{op}.
\end{eqnarray*}
Hence $(A^{op},\{\cdot,\cdot\}_{op},d_{op})$ is a $\lambda$-DL algebra.

Furthermore, we have
\begin{eqnarray*}
  \{a,b\cdot_{op}c\}_{op} &=& -\{a,b\cdot c\}  \\
  &\overset{\eqref{eq:Leibniz-rule}}{=}& -\{a,b\}\cdot c-b\cdot\{a,c\} \\
  &=& \{a,b\}_{op}\cdot_{op}c+b\cdot_{op}\{a,c\}_{op}.
\end{eqnarray*}
The modified differential case can be proved similarly. This completes the proof.
\end{proof}

\begin{pro}\label{pro:Poi-tensor-alg}
Let $(A,\cdot,\{\cdot,\cdot\}_A,d_A)$ and $(B,\ast,\{\cdot,\cdot\}_B,d_B)$ be two $\lambda$-DP algebras. For all $a,a' \in A$ and $b,b' \in B$, set
\begin{eqnarray*}
  (a\otimes b) \star (a'\otimes b') &:=& (a\cdot a')\otimes(b\ast b'), \\
  d(a\otimes b) &:=& d_A(a)\otimes b+a\otimes d_B(b)+\lambda d_A(a)\otimes d_B(b),\\
  \{a\otimes b, a'\otimes b'\}&:=& (a\cdot a')\otimes\{b,b'\}_B+\{a,a'\}_A\otimes (b\ast b'),
\end{eqnarray*}
then $(A\otimes B,\star,\{\cdot,\cdot\},d)$ is a $\lambda$-DP algebra.

\end{pro}

\begin{proof}
For simplicity, we omit $\cdot,\ast$ and all subscripts of $d_A,d_B,\{\cdot,\cdot\}_A,\{\cdot,\cdot\}_B$ without confusion.
For all $a\otimes b,a'\otimes b',a''\otimes b'' \in A\otimes B$, we have
\begin{eqnarray*}
  d\big((a\otimes b) \star (a'\otimes b')\big) &=& d(aa'\otimes bb') \\
  &=& d(aa')\otimes bb'+aa'\otimes d(bb')+\lambda d(aa')\otimes d(bb')\\
  &\overset{\eqref{eq:lambda}}{=}& \big(d(a)a'\otimes bb'+aa'\otimes d(b)b'+\lambda d(a)a'\otimes d(b)b'\big)\\
  &+& \big(ad(a')\otimes bb'+aa'\otimes bd(b')+\lambda ad(a')\otimes bd(b')\big)\\
  &+& \big(\lambda d(a)d(a')\otimes bb'+\lambda ad(a')\otimes d(b)b'+\lambda^2 d(a)d(a')\otimes d(b)b'\\
  &+& \lambda d(a)a'\otimes bd(b')+aa'\otimes\lambda d(b)d(b')+\lambda^2 d(a)a'\otimes d(b)d(b')\\
  &+& \lambda^2 d(a)d(a')\otimes bd(b')+\lambda^2 ad(a')\otimes d(b)d(b')+\lambda^3 d(a)d(a')\otimes d(b)d(b')\big)\\
  &=& d(a\otimes b)\star(a'\otimes b')+(a\otimes b)\star d(a'\otimes b')+\lambda d(a\otimes b)\star d(a'\otimes b'),
\end{eqnarray*}
thus, $(A\otimes B,\cdot,d)$ is a $\lambda$ differential algebra.

On the other hand, since $(A,\cdot,d_A)$ and $(B,\ast,d_B)$ are two $\lambda$-DL algebras, we obtain 
\begin{eqnarray*}
  \{a\otimes b,a'\otimes b'\} &=& aa'\otimes\{b,b'\}+\{a,a'\}\otimes bb'\\
  &=& -a'a\otimes\{b',b\}-\{a',a\}\otimes b'b=-\{a'\otimes b',a\otimes b\},
\end{eqnarray*}
\begin{eqnarray*}
  \{a\otimes b,\{a'\otimes b',a''\otimes b''\}\} &=& \{a\otimes b,a'a''\otimes\{b',b''\}\}+\{a\otimes b,\{a',a''\}\otimes b'b''\} \\
  &=& \{aa'\otimes\{b,b'\},a''\otimes b''\}+\{\{a,a'\}\otimes bb',a''\otimes b''\} \\
  &+& \{a'\otimes b',aa''\otimes \{b,b''\}\}+\{a'\otimes b',\{a,a''\}\otimes bb''\} \\
  &=& \{\{a\otimes b,a'\otimes b'\},a''\otimes b''\}+\{a'\otimes b',\{a\otimes b,a''\otimes b''\}\},
\end{eqnarray*}
and
\begin{eqnarray*}
  d\big(\{a\otimes b,a'\otimes b'\}\big) &=& d\big(aa'\otimes \{b,b'\}\big)+d\big(\{a,a'\}\otimes bb'\big) \\
  &=& d(aa')\otimes \{b,b'\}+aa'\otimes d\big(\{b,b'\}\big)+\lambda d(aa')\otimes d\big(\{b,b'\}\big) \\
  &+& d\big(\{a,a'\}\big)\otimes bb'+\{a,a'\}\otimes d(bb')+\lambda d\big(\{a,a'\}\big)\otimes d(bb')\\
  &\overset{\eqref{eq:lambda} \eqref{eq:diff-Liealg}}{=}& \big(d(a)a'\otimes \{b,b'\}+\{d(a),a'\}\otimes bb'\big)+\big(aa'\otimes\{d(b),b'\}+\{a,a'\}\otimes d(b)b'\big)\\
  &+&\big(\lambda d(a)a'\otimes\{d(b),b'\}+\{\lambda d(a),a'\}\otimes d(b)b'\big)\\
  &+&\big(ad(a')\otimes \{b,b'\}+\{a,d(a')\}\otimes bb'\big)+\big(aa'\otimes\{b,d(b')\}+\{a,a'\}\otimes bd(b')\big)\\
  &+&\big(a\lambda d(a')\otimes\{b,d(b')\}+\{a,\lambda d(a')\}\otimes bd(b')\big)\\
  &+&\lambda\big(d(a)d(a')\otimes \{b,b'\}+\{d(a),d(a')\}\otimes bb'\\
  &+&ad(a')\otimes\{d(b),b'\}+\{a,d(a')\}\otimes d(b)b'+\lambda d(a)d(a')\otimes\{d(b),b'\}\\
  &+&\{\lambda d(a),d(a')\}\otimes d(b)b'+d(a)a'\otimes\{b,d(b')\}+\{d(a),a'\}\otimes bd(b')\\
  &+&aa'\otimes\{d(b),d(b')\}+\{a,a'\}\otimes d(b)d(b')+\lambda d(a)a'\otimes\{d(b),d(b')\}\\
  &+&\{\lambda d(a),a'\}\otimes d(b)d(b')+d(a)\lambda d(a')\otimes\{b,d(b')\}+\{d(a),\lambda d(a')\}\otimes bd(b')\\
  &+&a\lambda d(a')\otimes\{d(b),d(b')\}+\{a,\lambda d(a')\}\otimes d(b)d(b')\\
  &+&\lambda d(a)\lambda d(a')\otimes\{d(b),d(b')\}+\{\lambda d(a),\lambda d(a')\}\otimes d(b)d(b')\big)\\
  &=&\{d(a)\otimes b,a'\otimes b'\}+\{a\otimes d(b),a'\otimes b'\}+\{\lambda d(a)\otimes d(b),a'\otimes b'\}\\
  &+&\{a\otimes b,d(a')\otimes b'\}+\{a\otimes b,a'\otimes d(b')\}+\{a\otimes b,\lambda d(a')\otimes d(b')\}\\
  &+&\lambda\{d(a)\otimes b+a\otimes d(b)+\lambda d(a)\otimes d(b),\\
  && d(a')\otimes b'+a'\otimes d(b')+\lambda d(a')\otimes d(b')\}\\
  &=&\{d(a\otimes b),a'\otimes b'\}+\{a\otimes b,d(a'\otimes b')\}+\lambda \{d(a\otimes b),d(a'\otimes b')\}.
\end{eqnarray*}
Hence $(A\otimes B,\{\cdot,\cdot\},d)$ is a $\lambda$-DL algebra.

Moreover, the Leibniz rule follows from
\begin{eqnarray*}
  \{a\otimes b,(a'\otimes b')\star(a''\otimes b'')\} &=& \{a\otimes b,a'a''\otimes b'b''\}\\
  &=&a(a'a'')\otimes\{b,b'b''\}+\{a,a'a''\}\otimes b(b'b'') \\
  &\overset{\eqref{eq:Leibniz-rule}}{=}& \big((aa')a''\otimes\{b,b'\}b''+\{a,a'\}a''\otimes (bb')b''\big)\\
  &+& \big(a'(aa'')\otimes b'\{b,b''\}+a'\{a,a''\}\otimes b'(bb'')\big)\\
  &=&\big(aa'\otimes\{b,b'\}+\{a,a'\}\otimes bb'\big)\star(a''\otimes b'')\\
  &+&(a'\otimes b')\star\big(aa''\otimes\{b,b''\}+\{a,a''\}\otimes bb'\big)\\
  &=&\{a\otimes b,a'\otimes b'\}\star(a''\otimes b'')+(a'\otimes b')\star\{a\otimes b,a''\otimes b''\}.
\end{eqnarray*}
The proof is completed.
\end{proof}

\begin{pro}\label{pro:mPoi-tensor-alg}
Let $(A,\cdot,\{\cdot,\cdot\}_A,d_A)$ and $(B,\ast,\{\cdot,\cdot\}_B,d_B)$ be two $\lambda$-mDP algebras, then
$(A\otimes B,\star,\{\cdot,\cdot\},d)$ is a $\lambda$-mDP algebra, where for all $a,a' \in A$ and $b,b' \in B$,
\begin{eqnarray*}
  (a\otimes b) \star (a'\otimes b') &:=& (a\cdot a')\otimes(b\ast b'), \\
  d(a\otimes b) &:=& d_A(a)\otimes b+a\otimes d_B(b)+\lambda a\otimes b,\\
  \{a\otimes b, a'\otimes b'\}&:=& (a\cdot a')\otimes\{b,b'\}_B+\{a,a'\}_A\otimes (b\ast b').
\end{eqnarray*}
\end{pro}

\begin{proof}
It is straightforward by Proposition \ref{pro:Poi-tensor-alg}.
\end{proof}

\begin{pro}\label{pro:Poi-quotient-alg}
Let $(A,\cdot,\{\cdot,\cdot\},d)$ be a $\lambda$-DP (or $\lambda$-mDP) algebra, and $B\lhd A$. For all $a,b \in A$, we define a $\lambda$-DP (or $\lambda$-mDP) algebra structure on the quotient algebra $A/B$ as follows:
\begin{eqnarray*}
  (a+B) \cdot_{A/B} (b+B) &:=& a \cdot b+B,\\
  \{a+B,b+B\}_{A/B} &:=& \{a,b\}+B, \\
  d_{A/B}(a+B) &:=& d(a)+B.
\end{eqnarray*}
Then the quotient algebra $(A/B,\cdot_{A/B},\{\cdot,\cdot\}_{A/B},d_{A/B})$ is a $\lambda$-DP (or $\lambda$-mDP) algebra.
\end{pro}

\begin{proof}
For all $a, b \in A$, we have 
\begin{eqnarray*}
  d_{A/B}\big((a+B) \cdot_{A/B} (b+B)\big) &=& d_{A/B}(a\cdot b+B)=d(a\cdot b)+B\\
  &\overset{\eqref{eq:lambda}}{=}&d(a)\cdot b+a\cdot d(b)+\lambda d(a)\cdot d(b)+B\\
  &=& d_{A/B}(a+B)\cdot_{A/B}(b+B)+(a+B)\cdot_{A/B}d_{A/B}(b+B)\\
  &+&\lambda d_{A/B}(a+B)\cdot_{A/B}d_{A/B}(b+B),
\end{eqnarray*}
and
\begin{eqnarray*}
  d_{A/B}\big(\{a+B,b+B\}_{A/B}\big) &=& d_{A/B}\big(\{a,b\}+B\big)=d\big(\{a,b\}\big)+B \\
  &\overset{\eqref{eq:diff-Liealg}}{=}& \{d(a),b\}+\{a,d(b)\}+\lambda\{d(a),d(b)\}+B \\
  &=& \{d_{A/B}(a+B),b+B\}_{A/B}+\{a+B,d_{A/B}(b+B)\}_{A/B}\\
  &+&\lambda\{d_{A/B}(a+B),d_{A/B}(b+B)\}_{A/B}.
\end{eqnarray*}
Furthermore, let $c \in A$, we have
\begin{eqnarray*}
  \{a+B,(b+B)\cdot_{A/B}(c+B)\}_{A/B} &=& \{a+B,b\cdot c+B\}_{A/B}=\{a,b\cdot c\}+B\\
  &\overset{\eqref{eq:Leibniz-rule}}{=}&\{a,b\}\cdot c+b\cdot\{a,c\}+B \\
  &=&  \{a+B,b+B\}_{A/B}\cdot_{A/B}(c+B)+(b+B)\cdot_{A/B}\{a+B,c+B\}_{A/B}.
\end{eqnarray*}
The modified differential case can be proved similarly. 
\end{proof}

\begin{pro}\label{pro:i-iii}
Let $(A,\cdot_A,\{\cdot,\cdot\}_A,d_A)$ and $(B,\cdot_B,\{\cdot,\cdot\}_B,d_B)$ be two $\lambda$-DP (or $\lambda$-mDP) algebras and $f:A\rightarrow B$ be a $\lambda$-DP (or $\lambda$-mDP) algebra homomorphism. Then we have
\begin{description}
  \item[(i)] \label{i}$\mathrm{ker}f\lhd A,$
  \item[(ii)] $\mathrm{im}f<B,$
  \item[(iii)] $A/\mathrm{ker}f\cong \mathrm{im}f$ as $\lambda$-DP (or $\lambda$-mDP) algebras$.$
\end{description}

\end{pro}

\begin{proof}
(1)~Since $f\big(d_A(\mathrm{ker}f)\big)=d_B(f(\mathrm{ker}f))=0$, we have $d_A(\mathrm{ker}f) \subseteq \mathrm{ker}f$. 

Note
\begin{equation*}
  f(a \cdot_A \mathrm{ker}f)=f(a) \cdot_B f(\mathrm{ker}f)=0,~~f(\mathrm{ker}f \cdot_A a)=f(\mathrm{ker}f) \cdot_B f(a)=0,\quad\forall a \in A.
\end{equation*}
Thus, we have $a \cdot_A \mathrm{ker}f, \mathrm{ker}f \cdot_A a \subseteq \mathrm{ker}f$.

Furthermore, we have 
\begin{equation*}
  f\big(\{a,\mathrm{ker}f\}_A\big)=\{f(a),f(\mathrm{ker}f)\}_B=0,~~f\big(\{\mathrm{ker}f,a\}_A\big)=\{f(\mathrm{ker}f),f(a)\}_B=0,\quad\forall a \in A.
\end{equation*}
Therefore, we derive $\{a,\mathrm{ker}f\}_A, \{\mathrm{ker}f,a\}_A \subseteq \mathrm{ker}f$.

(2)~For $y \in \mathrm{im}f$, there exists some $a \in A$ such that $f(a)=y$. Then $d_B(y)=d_B\big(f(a)\big)=f\big(d_A(a)\big) \in \mathrm{im}f$. Similarly, for $z \in \mathrm{im}f$, we have $f(b)=z$ for some $b \in A$. Thus, we have $\{y,z\}_B=\{f(a),f(b)\}_B=f\big(\{a,b\}_A\big) \in \mathrm{im}f$.

(3)~By Proposition \ref{pro:Poi-quotient-alg} and (i), we obtain $A/\mathrm{ker}f$ is a $\lambda$-DP ($\lambda$-mDP) algebra. Define a linear map
\begin{equation*}
  \phi:A/\mathrm{ker}f\rightarrow \mathrm{im}f \quad by
\end{equation*}
\begin{equation*}
  a+\mathrm{ker}f\mapsto f(a),\quad\forall a \in A.
\end{equation*}
By the first isomorphism theorem, we obtain that $\phi$ is a $\lambda$-DP ($\lambda$-mDP) algebra isomorphism. For all $a,b \in A$, we have
\begin{eqnarray*}
  \phi\big(d_{A/\mathrm{ker}f}(a+\mathrm{ker}f)\big) &=& f\big(d_{A/\mathrm{ker}f}(a+\mathrm{ker}f)\big)=f\big(d_A(a)+\mathrm{ker}f\big) \\
  &=& f\big(d_A(a)\big)=d_B\big(f(a)\big)=d_B\big(\phi(a+\mathrm{ker}f)\big).
\end{eqnarray*}
Furthermore, we have
\begin{eqnarray*}
  \phi\big(\{a+\mathrm{ker}f,b+\mathrm{ker}f\}_{A/\mathrm{ker}f}\big) &=& f\big(\{a+\mathrm{ker}f,b+\mathrm{ker}f\}_{A/\mathrm{ker}f}\big)=f\big(\{a,b\}_A+\mathrm{ker}f\big) \\
  &=& f\big(\{a,b\}_A\big)=\{f(a),f(b)\}_B=\{\phi(a+\mathrm{ker}f),\phi(b+\mathrm{ker}f)\}_B.
\end{eqnarray*}
Thus, we deduce $\phi$ is compatible with $d_{A/\mathrm{ker}f}$ and $\{\cdot,\cdot\}_{A/\mathrm{ker}f}$.
\end{proof}

Recall that a complex \cite{homology} in an abelian category $\mathcal{A}$ is a sequence of morphisms (called differentials),
\begin{equation*}
  (\mathbf{C_{\bullet}},d_{\bullet})=\longrightarrow A_{n+1}\stackrel{d_{n+1}}{\longrightarrow} A_n\stackrel{d_n}{\longrightarrow} A_{n-1}\longrightarrow,
\end{equation*}
such that 
\begin{equation*}
  d_nd_{n+1}=0, \quad\forall n \in \mathbb{Z}.
\end{equation*}
We write $\mathbf{C_{\bullet}}$ instead of $(\mathbf{C_{\bullet}},d_{\bullet})$ briefly and the category of all complexes in $\mathcal{A}$ is denoted by $\mathbf{Comp}(\mathcal{A})$. If $\mathbf{C_{\bullet}}$ is a complex in $\mathbf{Comp}(\mathcal{A})$ and $n \in \mathbb{Z}$, then its $n$-th homology is
\begin{equation*}
  \mathcal{H}_n(\mathbf{C_{\bullet}})=\mathrm{ker}d_n/\mathrm{im}d_{n+1}.
\end{equation*}

\begin{pro}\label{pro:coho-alg}
Let $(A,\cdot,\{\cdot,\cdot\},d)$ be a $\lambda$-DP (or $\lambda$-mDP) algebra. For all $x,y\in A$, we define a $\lambda$-DP (or $\lambda$-mDP) algebra structure on cohomology algebra $\big(\mathcal{H}(A),\star,\{\cdot,\cdot\}_{\mathcal{H}(A)},d_{\mathcal{H}(A)}\big)$ of $(A,\cdot,\{\cdot,\cdot\},d)$ by
\begin{eqnarray*}
  (x+\mathrm{im}d) \star (y+\mathrm{im}d) &:=& x\cdot y+\mathrm{im}d, \\
   \{x+\mathrm{im}d,y+\mathrm{im}d\}_{\mathcal{H}(A)}&:=&\{x,y\}+\mathrm{im}d,\\
  d_{\mathcal{H}(A)}(x+\mathrm{im}d) &:=& d(x)+\mathrm{im}d.
\end{eqnarray*}
Moreover, if $(B,\cdot_B,\{\cdot,\cdot\}_B,d_B)$ is another $\lambda$-DP (or $\lambda$-mDP) algebra such that $A\cong B$ as $\lambda$-DP (or $\lambda$-mDP) algebras, then $\mathcal{H}(A)\cong \mathcal{H}(B)$ as $\lambda$-DP (or $\lambda$-mDP) algebras.

\end{pro}

\begin{proof}
By \eqref{eq:lambda}, for all $x,y \in A$, we have
\begin{eqnarray*}
  d_{\mathcal{H}(A)}\big((x+\mathrm{im}d)\star (y+\mathrm{im}d)\big) &=& d_{\mathcal{H}(A)}(x\cdot y+\mathrm{im}d)=d(x\cdot y)+\mathrm{im}d \\
  &\overset{\eqref{eq:lambda}}{=}& d(x)\cdot y+x\cdot d(y)+\lambda d(x)\cdot d(y)+\mathrm{im}d \\
  &=& \big(d(x)+\mathrm{im}d\big) \star (y+\mathrm{im}d)+(x+\mathrm{im}d)\star \big(d(y)+\mathrm{im}d\big)\\
  &+& \lambda\big(d(x)+\mathrm{im}d\big)\star\big(d(y)+\mathrm{im}d\big)\\
  &=& \big(d_{\mathcal{H}(A)}(x+\mathrm{im}d)\big)\star (y+\mathrm{im}d)+(x+\mathrm{im}d)\star \big(d_{\mathcal{H}(A)}(y+\mathrm{im}d)\big)\\
  &+& \lambda\big(d_{\mathcal{H}(A)}(x+\mathrm{im}d)\big)\star\big(d_{\mathcal{H}(A)}(y+\mathrm{im}d)\big).
\end{eqnarray*}
Thus, $(\mathcal{H}(A),\star,d_{\mathcal{H}(A)})$ is a $\lambda$-differential algebra. 

Since $(A,\{\cdot,\cdot\})$ is a Lie algebra, we have
\begin{equation*}
  \{x+\mathrm{im}d,y+\mathrm{im}d\}_{\mathcal{H}(A)} = \{x,y\}+\mathrm{im}d=-\{y,x\}+\mathrm{im}d=-\{y+\mathrm{im}d,x+\mathrm{im}d\}_{\mathcal{H}(A)},
\end{equation*}
and
\begin{eqnarray*}
  \{x+\mathrm{im}d,\{y+\mathrm{im}d,z+\mathrm{im}d\}_{\mathcal{H}(A)}\}_{\mathcal{H}(A)} &=& \{x+\mathrm{im}d,\{y,z\}+\mathrm{im}d\}_{\mathcal{H}(A)} \\
  &=& \{x,\{y,z\}\}+\mathrm{im}d=\{\{x,y\},z\}+\{y,\{x,z\}\}+\mathrm{im}d \\
  &=& \{\{x,y\}+\mathrm{im}d,z+\mathrm{im}d\}_{\mathcal{H}(A)}+\{y+\mathrm{im}d,\{x,z\}+\mathrm{im}d\}_{\mathcal{H}(A)} \\
  &=& \{\{x+\mathrm{im}d,y+\mathrm{im}d\}_{\mathcal{H}(A)},z+\mathrm{im}d\}_{\mathcal{H}(A)}\\
  &+&\{y+\mathrm{im}d,\{x+\mathrm{im}d,z+\mathrm{im}d\}_{\mathcal{H}(A)}\}_{\mathcal{H}(A)}.
\end{eqnarray*}
Furthermore, by \eqref{eq:diff-Liealg}, the following equation holds.
\begin{eqnarray*}
  d_{\mathcal{H}(A)}\big(\{x+\mathrm{im}d,y+\mathrm{im}d\}_{\mathcal{H}(A)}\big) &=& d_{\mathcal{H}(A)}\big(\{x,y\}+\mathrm{im}d\big)=d\big(\{x,y\}\big)+\mathrm{im}d \\
  &\overset{\eqref{eq:diff-Liealg}}{=}& \{d(x),y\}+\{x,d(y)\}+\lambda\{d(x),d(y)\}+\mathrm{im}d \\
  &=& \{d(x)+\mathrm{im}d,y+\mathrm{im}d\}_{\mathcal{H}(A)}+\{x+\mathrm{im}d,d(y)+\mathrm{im}d\}_{\mathcal{H}(A)}\\
  &+& \lambda\{d(x)+\mathrm{im}d,d(y)+\mathrm{im}d\}_{\mathcal{H}(A)}\\
  &=& \{d_{\mathcal{H}(A)}(x+\mathrm{im}d),y+\mathrm{im}d\}_{\mathcal{H}(A)}
  +\{x+\mathrm{im}d,d_{\mathcal{H}(A)}(y+\mathrm{im}d)\}_{\mathcal{H}(A)}\\
  &+& \lambda\{d_{\mathcal{H}(A)}(x+\mathrm{im}d),d_{\mathcal{H}(A)}(y+\mathrm{im}d)\}_{\mathcal{H}(A)}.
\end{eqnarray*}
Applying \eqref{eq:Leibniz-rule} to $(A,\cdot,\{\cdot,\cdot\})$ gives
\begin{eqnarray*}
  \{x+\mathrm{im}d,(y+\mathrm{im}d)\star(z+\mathrm{im}d)\}_{\mathcal{H}(A)} &=& \{x+\mathrm{im}d,y\cdot z+\mathrm{im}d\}_{\mathcal{H}(A)} \\
  &=& \{x,y\cdot z\}+\mathrm{im}d\overset{\eqref{eq:Leibniz-rule}}{=}\{x,y\}\cdot z+y\cdot\{x,z\}+\mathrm{im}d\\
  &=& \big(\{x,y\}+\mathrm{im}d\big)\star (z+\mathrm{im}d)+(y+\mathrm{im}d)\star\big(\{x,z\}+\mathrm{im}d\big)\\
  &=& \{x+\mathrm{im}d,y+\mathrm{im}d\}_{\mathcal{H}(A)}\star (z+\mathrm{im}d)\\
  &+& (y+\mathrm{im}d)\star\{x+\mathrm{im}d,z+\mathrm{im}d\}_{\mathcal{H}(A)}.
\end{eqnarray*}
Therefore, $(\mathcal{H}(A),\star,\{\cdot,\cdot\}_\mathcal{H}(A),d_\mathcal{H}(A))$ is a $\lambda$-DP algebra.

On the other hand, let $f:A\rightarrow B$ be a $\lambda$-DP algebra isomorphism. Define a linear map $\widetilde{f}:\mathcal{H}(A)\rightarrow \mathcal{H}(B)$ by
\begin{equation*}
  \widetilde{f}(a+\mathrm{im}d_A)=f(a)+\mathrm{im}d_B, \quad\forall a+\mathrm{im}d_A \in \mathcal{H}(A).
\end{equation*}
By the definitions of a $\lambda$-DP algebra isomorphism, we immediately obtain that such $\widetilde{f}$ is well defined and is a bijective $\lambda$-DP algebra homomorphism. 

The modified differential case can be proved similarly. Hence the conclusion follows.

\end{proof}

\subsection{The modules of $\lambda$-DP algebras and $\lambda$-mDP algebras}\label{subsec-DP module}
$\\$

Let $(A,[\cdot,\cdot],d)$ be a $\lambda$-DL algebra. A left $\lambda$-differential Lie $A$-module \cite{Unen-dLie} is a triple $(M,[\cdot,\cdot]_M,d_M)$ if $M$ is a left $A$-module, $[\cdot,\cdot]_M: A\otimes M\rightarrow M$ and $d_M:M\rightarrow M$ are two linear maps such that
\begin{eqnarray}
  \label{eq:diff-Liemodu}d_M([a,m]_M)&=&[d(a),m]_M+[a,d_M(m)]_M+\lambda[d(a),d_M(m)]_M,\\
  \nonumber[a,[b,m]_M]_M &=&[[a,b],m]_M+[b,[a,m]_M]_M,\quad\forall a,b \in A,~m \in M.
\end{eqnarray}
If \eqref{eq:diff-Liemodu} is substituted for
\begin{equation*}\label{eq:mo-diff-Liemodu}
  d_M([a,m]_M)=[d(a),m]_M+[a,d_M(m)]_M+\lambda[a,m]_M,\quad\forall a,b \in A,~m \in M,
\end{equation*}
then the triple $(M,[\cdot,\cdot]_M,d_M)$ is called a left modified $\lambda$-differential Lie $A$-module. 

Right (modified) $\lambda$-differential Lie $A$-modules can be defined in the similar manner. We denote left (resp. right) $\lambda$-differential Lie $A$-module by {\bf $A_{\lambda-DL}$} (resp. {\bf $_{\lambda-DL}A$}). For left (resp. right) modified $\lambda$-differential Lie $A$-module, we denote by {\bf $A_{\lambda-mDL}$} (resp. {\bf $_{\lambda-mDL}A$}).

\begin{defi}\label{defi:diff-Poimodule}
Let $(A,\cdot,\{\cdot,\cdot\},d)$ be a $\lambda$-DP algebra. A left {\bf $\lambda$-differential Poisson $A$-module} is a quadruple $(M,\ast,\{\cdot,\cdot\}_M,d_M)$ if $(M,\ast,d_M)$ is a left $\lambda$-differential $A$-module, $(M,\{\cdot,\cdot\}_M,d_M) \in A_{\lambda-DL}$ such that for any $a,b \in A, m \in M$,
\begin{eqnarray}
  \label{eq:diff-Poimodule1}\{a,b\ast m\}_M &=& \{a,b\}\ast m+b\ast\{a,m\}_M, \\
  \label{eq:diff-Poimodule2}\{a\cdot b,m\}_M &=& a\ast\{b,m\}_M+b\ast\{a,m\}_M.
\end{eqnarray}

\end{defi}
If $(A,\cdot,\{\cdot,\cdot\},d)$ is a $\lambda$-mDP algebra, $(M,\ast,d_M)$ is a left modified $\lambda$-differential $A$-module, \\
$(M,\{\cdot,\cdot\}_M,d_M) \in A_{\lambda-mDL}$, then the quadruple $(M,\ast,\{\cdot,\cdot\}_M,d_M)$ satisfying \eqref{eq:diff-Poimodule1} \eqref{eq:diff-Poimodule2} is a left {\bf modified $\lambda$-differential Poisson $A$-module}. 

For abbreviation, we denote left (resp. right) $\lambda$-differential Poisson $A$-module by {\bf $A_{\lambda-DP}$} (resp. {\bf $_{\lambda-DP}A$}), and denote left (resp. right) modified $\lambda$-differential Poisson $A$-module by {\bf $A_{\lambda-mDP}$} (resp. {\bf $_{\lambda-mDP}A$}).

\begin{defi}
Let $(A,\cdot,\{\cdot,\cdot\},d)$ be a $\lambda$-DP (or $\lambda$-mDP) algebra, $(M,\ast,\{\cdot,\cdot\}_M,d_M) \in A_{\lambda-DP} \\ (or~A_{\lambda-mDP})$ and $(N,\ast)$ be a left $A$-sub-module of $(M,\ast)$. If $d_M(N)\subseteq N$ and $\{A,N\}_M\subseteq N$, 
then $(N,\ast,\{\cdot,\cdot\}_M,d_M)$ is called a left $\lambda$-DP (or $\lambda$-mDP) $A$-sub-module of $(M,\ast,\{\cdot,\cdot\}_M,d_M)$, 
denoted by $N\leq_{\mathcal{P}} M$.
\end{defi}
Right $\lambda$-DP (or $\lambda$-mDP) $A$-sub-modules can be defined in the similar manner. For any $\lambda$-DP (or $\lambda$-mDP) $A$-module $M$, we have $0\leq_{\mathcal{P}} M$ and $M\leq_{\mathcal{P}} M$. If $0$ and $M$ are the only $\lambda$-DP (or $\lambda$-mDP) $A$-sub-modules of $M$, then we call $M$ a simple $\lambda$-DP (or $\lambda$-mDP) $A$-module.

\begin{defi}
Let $(A,\cdot,\{\cdot,\cdot\},d)$ be a $\lambda$-DP (or $\lambda$-mDP) algebra, and $(M,\ast_M,\{\cdot,\cdot\}_M,d_M), \\
(N,\ast_N,\{\cdot,\cdot\}_N,d_N) \in A_{\lambda-DP}~(or~A_{\lambda-mDP})$. A left $A$-module homomorphism  $f:M\rightarrow N$ is called a left $\lambda$-DP (or $\lambda$-mDP) $A$-module {\bf homomorphism} if $f(\{a,m\}_M)=\{a,f(m)\}_N$ and $f\circ d_M=d_N\circ f$ for all $a \in A, m \in M$. Moreover, $f$ is called a left $\lambda$-DP (or $\lambda$-mDP) $A$-module {\bf monomorphism} if $\mathrm{Ker}f=0$, a left $\lambda$-DP (or $\lambda$-mDP) $A$-module {\bf epimorphism} if $\mathrm{Im}f=N$, a left $\lambda$-DP (or $\lambda$-mDP) $A$-module {\bf isomorphism} if it is a bijection, denoted by $M\cong N$.
\end{defi}

Propositions \ref{pro:oppo-module} $\sim$ \ref{pro:coho-module} are straightforward applications of Propositions \ref{pro:oppo-alg} $\sim$ \ref{pro:coho-alg}.

\begin{pro}\label{pro:oppo-module}
Let $(A,\cdot,\{\cdot,\cdot\},d)$ be a $\lambda$-DP (or $\lambda$-mDP) algebra, and $(M,\ast,\{\cdot,\cdot\}_M,d_M) \in A_{\lambda-DP} \\ (or~A_{\lambda-mDP})$. Then $(M,\ast^{op},\{\cdot,\cdot\}_M^{op},d_M^{op}) \in _{\lambda-DP}A~(or~_{\lambda-mDP}A)$, where for all $a \in A, m \in M$,
\begin{eqnarray*}
  m\ast^{op}a &:=& a\ast m, \\
  \{m,a\}_M^{op} &:=& \{a,m\}_M, \\
  d_M^{op} &:=& d_M.
\end{eqnarray*}
\end{pro}

\begin{pro}\label{pro:Poi-tensor-modu}
Let $(A,\cdot_A,\{\cdot,\cdot\}_A,d_A)$,~$(B,\cdot_B,\{\cdot,\cdot\}_B,d_B)$ be two $\lambda$-DP algebras. For any \\
$(M,\ast_M,\{\cdot,\cdot\}_M,d_M) \in A_{\lambda-DP}$ and $(N,\ast_N,\{\cdot,\cdot\}_N,d_N) \in B_{\lambda-DP}$, there is a left $\lambda$-DP $(A\otimes B)$-module structure on $M\otimes N$ defined by
\begin{eqnarray*}
  (a\otimes b)\star (m\otimes n) &:=& (a\ast_Mm) \otimes (b\ast_Nn), \\
  d_{M\otimes N}(m\otimes n) &:=& d_M(m)\otimes n+m\otimes d_N(n)+\lambda d_M(m)\otimes d_N(n), \\
  \{a\otimes b,m\otimes n\}_{M\otimes N} &:=& (a\ast_Mm)\otimes \{b,n\}_N+\{a,m\}_M\otimes (b\ast_Nn),
\end{eqnarray*}
for all $a \in A, b \in B, m \in M, n \in N$.
\end{pro}


\begin{pro}\label{pro:mPoi-tensor-modu}
Let $(A,\cdot_A,\{\cdot,\cdot\}_A,d_A)$,~$(B,\cdot_B,\{\cdot,\cdot\}_B,d_B)$ be two $\lambda$-mDP algebras. If \\
$(M,\ast_M,\{\cdot,\cdot\}_M,d_M) \in A_{\lambda-mDP}$ and $(N,\ast_N,\{\cdot,\cdot\}_N,d_N) \in B_{\lambda-mDP}$, then $(M\otimes N,\star,\{\cdot,\cdot\}_{M\otimes N},d_{M\otimes N}) \in (A\otimes B)_{\lambda-mDP}$, where for all $a \in A, b \in B, m \in M, n \in N$,
\begin{eqnarray*}
  (a\otimes b)\star (m\otimes n) &:=& (a\ast_Mm) \otimes (b\ast_Nn), \\
  d_{M\otimes N}(m\otimes n) &:=& d_M(m)\otimes n+m\otimes d_N(n)+\lambda m\otimes n, \\
  \{a\otimes b,m\otimes n\}_{M\otimes N} &:=& (a\ast_Mm)\otimes \{b,n\}_N+\{a,m\}_M\otimes (b\ast_Nn).
\end{eqnarray*}

\end{pro}

\begin{pro}\label{pro:Poi-quotient-modu}
Let $(A,\cdot,\{\cdot,\cdot\},d)$ be a $\lambda$-DP (or $\lambda$-mDP) algebra, $(M,\ast,\{\cdot,\cdot\}_M,d_M) \in A_{\lambda-DP} \\ (or~A_{\lambda-mDP})$, and $N\leq_{\mathcal{P}} M$. For all $a \in A, m \in M$, we define a left $\lambda$-DP (or $\lambda$-mDP) $A$-module structure on the quotient module $M/N$ as follows:
\begin{eqnarray*}
  a\ast_{M/N}(m+N) &:=& a\ast m+N, \\
  \{a,m+N\}_{M/N} &:=& \{a,m\}_M+N, \\
  d_{M/N}(m+N) &:=& d_M(m)+N.
\end{eqnarray*}

\end{pro}

\begin{pro}\label{pro:i'-iii'}
Let $(A,\cdot,\{\cdot,\cdot\},d)$ be a $\lambda$-DP (or $\lambda$-mDP) algebra. If $(M,\ast_M,\{\cdot,\cdot\}_M,d_M), \\
(N,\ast_N,\{\cdot,\cdot\}_N,d_N) \in A_{\lambda-DP}~(or~A_{\lambda-mDP})$, and $f:M\rightarrow N$ is a left $\lambda$-DP (or $\lambda$-mDP) $A$-module homomorphism. Then
\begin{description}
  \item[(i)] $\mathrm{ker}f\leq_{\mathcal{P}} M,$
  \item[(ii)]  if $X\leq_{\mathcal{P}} M$, then $f(X)\leq_{\mathcal{P}} N$. Particularly, $\mathrm{im}f\leq_{\mathcal{P}} N,$
  \item[(iii)] $M/\mathrm{ker}f\cong \mathrm{im}f$ as left $\lambda$-DP (or $\lambda$-mDP) $A$-modules $.$
\end{description}
\end{pro}

\begin{pro}\label{pro:coho-module}
Let $(A,\cdot,\{\cdot,\cdot\},d)$ be a $\lambda$-DP (or $\lambda$-mDP) algebra and $(M,\ast_M,\{\cdot,\cdot\}_M,d_M) \in A_{\lambda-DP} 
~(or~A_{\lambda-mDP})$. For all $a\in A, m\in M$, we define a $\lambda$-DP (or $\lambda$-mDP) $\mathcal{H}(A)$-module structure on cohomology algebra $(\mathcal{H}(M),\star,\{\cdot,\cdot\}_{\mathcal{H}(M)},d_{\mathcal{H}(M)})$ as follows:
\begin{eqnarray*}
  (a+\mathrm{im}d)\star(m+\mathrm{im}d_M) &:=& a\ast_M m+\mathrm{im}d_M, \\
  \{a+\mathrm{im}d,m+\mathrm{im}d_M\}_{\mathcal{H}(M)}&:=&\{a,m\}_M+\mathrm{im}d_M,\\
  d_{\mathcal{H}(M)}(m+\mathrm{im}d_M)&:=& d_M(m)+\mathrm{im}d_M.
\end{eqnarray*}
Moreover, if $(N,\ast_N,\{\cdot,\cdot\}_N,d_N) \in A_{\lambda-DP}~(or~A_{\lambda-mDP})$ and $M\cong N$, then $\mathcal{H}(M)\cong\mathcal{H}(N)$ as \\ $\mathcal{H}(A)_{\lambda-DP}~(or~\mathcal{H}(A)_{\lambda-mDP})$.

\end{pro}

\section{Universal enveloping algebras of $\lambda$-DP (or $\lambda$-mDP) algebras}\label{sec:Uea-DPa}
In this section, we generalize the definitions of universal enveloping algebras to $\lambda$-DP (or $\lambda$-mDP) algebras by generators and defining relations, see \cite{Umirbaev}. 

\begin{defi}\label{def:Uni-enve-alg}
Let $(A,\cdot,\{\cdot,\cdot\},d)$ be a $\lambda$-DP (or $\lambda$-mDP) algebra, and
\begin{equation*}
  M_A=\{M_a:a \in A\},\quad H_A=\{H_a:a \in A\}
\end{equation*}
are two copies of the linear space $A$ endowed with two $\mathbf{k}$-linear isomorphisms $M:A\rightarrow M_A$, $H:A\rightarrow H_A$ sending $a$ to $M_a, H_a$ respectively. The {\bf universal enveloping algebra} $A^{ue}$ of $A$ is defined to be the quotient algebra of the free algebra generated by $M_A$ and $H_A$, subject to the following relations:
\begin{eqnarray}
  \label{eq:rela-M-pro}M_{ab} &=& M_{a}M_{b}, \\
  \label{eq:rela-H-bra}H_{\{a,b\}} &=& H_{a}H_{b}-H_{b}H_{a}, \\
  \label{eq:rela-M-bra}M_{\{a,b\}}&=& H_{a}M_{b}-M_{b}H_{a}, \\
  \label{eq:rela-H-pro}H_{ab} &=& M_{a}H_{b}+M_{b}H_{a}, \\
  \label{eq:rela-M-unit}M_1&=& 1,
\end{eqnarray}
for all $a,b \in A$.

\end{defi}

\begin{pro}\label{pro:uni-enve-alg}
Let $(A,\cdot,\{\cdot,\cdot\},d)$ be a $\lambda$-DP (or $\lambda$-mDP) algebra, then the universal enveloping algebra $A^{ue}$ has a (modified) $\lambda$-differential algebra structure $D:A^{ue}\rightarrow A^{ue}$ such that for all $a \in A$,
\begin{equation*}
  D(M_a)=M_{d(a)},\quad D(H_a)=H_{d(a)}.
\end{equation*}
\end{pro}

\begin{proof}
It suffices to prove that $D$ preserves all the relations \eqref{eq:rela-M-pro}$\sim$\eqref{eq:rela-M-unit}. It is obvious that $D(M_1)=D(1)$.
Let $a,b\in A$, then by \eqref{eq:lambda} \eqref{eq:rela-M-pro} \eqref{eq:rela-H-pro} we have
\begin{eqnarray*}
  D(M_{ab}) &=& M_{d(ab)}\overset{\eqref{eq:lambda}}{=}M_{d(a)b}+M_{ad(b)}+M_{\lambda d(a)d(b)} \\
  &\overset{\eqref{eq:rela-M-pro}}{=}& M_{d(a)}M_b+M_aM_{d(b)}+\lambda M_{d(a)}M_{d(b)} \\
  &=& D(M_a)M_b+M_aD(M_b)+\lambda D(M_a)D(M_b) \\
  &=& D(M_aM_b),
\end{eqnarray*}
and
\begin{eqnarray*}
  D(H_{ab}) &=& H_{d(ab)}\overset{\eqref{eq:lambda}}{=}H_{d(a)b}+H_{ad(b)}+H_{\lambda d(a)d(b)} \\
  &\overset{\eqref{eq:rela-H-pro}}{=}& \big(M_{d(a)}H_b+M_bH_{d(a)}\big)+\big(M_aH_{d(b)}+M_{d(b)}H_a\big)
  +\lambda\big(M_{d(a)}H_{d(b)}+M_{d(b)}H_{d(a)}\big)  \\
  &=& \big(D(M_a)H_b+M_aD(H_b)+\lambda D(M_a)D(H_b)\big)\\
  &+&\big(D(M_b)H_a+M_bD(H_a)+\lambda D(M_b)D(H_a)\big)\\
  &=& D(M_{a}H_{b}+M_{b}H_{a}).
\end{eqnarray*}
Furthermore, by \eqref{eq:diff-Liealg} \eqref{eq:rela-H-bra} \eqref{eq:rela-M-bra} we obtain
\begin{eqnarray*}
  D(H_{\{a,b\}}) &=& H_{d(\{a,b\})}\overset{\eqref{eq:diff-Liealg}}{=}H_{\{d(a),b\}}+H_{\{a,d(b)\}}+H_{\lambda\{d(a),d(b)\}} \\
  &\overset{\eqref{eq:rela-H-bra}}{=}& \big(H_{d(a)}H_b-H_bH_{d(a)}\big)+\big(H_aH_{d(b)}-H_{d(b)}H_a\big)+\lambda\big(H_{d(a)}H_{d(b)}-H_{d(b)}H_{d(a)}\big) \\
  &=& \big(D(H_a)H_b+H_aD(H_b)+\lambda H_{d(a)}H_{d(b)}\big)-\big(D(H_b)H_a+H_bD(H_a)+\lambda H_{d(b)}H_{d(a)}\big) \\
  &=& D(H_{a}H_{b}-H_{b}H_{a}),
\end{eqnarray*}
and
\begin{eqnarray*}
  D(M_{\{a,b\}}) &=& M_{d(\{a,b\})}=\overset{\eqref{eq:diff-Liealg}}{=}M_{\{d(a),b\}}+M_{\{a,d(b)\}}+M_{\lambda\{d(a),d(b)\}} \\
  &\overset{\eqref{eq:rela-M-bra}}{=}& \big(H_{d(a)}M_b-M_bH_{d(a)}\big)+\big(H_aM_{d(b)}-M_{d(b)}H_a\big)
  +\lambda\big(H_{d(a)}M_{d(b)}-M_{d(b)}H_{d(a)}\big) \\
  &=& \big(D(H_a)M_b+H_aD(M_b)+\lambda D(H_a)D(M_b)\big)\\
  &-&\big(D(M_b)H_a+M_bD(H_a)+\lambda D(M_b)D(H_a)\big) \\
  &=& D(H_{a}M_{b}-M_{b}H_{a}).
\end{eqnarray*}
The modified differential case can be proved similarly. Hence, the conclusion holds.

\end{proof}

\begin{defi}\label{defi:P-triple}
Let $(A,\cdot_A,\{\cdot,\cdot\}_A,d_A)$ be a $\lambda$-DP (or $\lambda$-mDP) algebra. For all $a,b \in A$, if
\begin{itemize}
  \item $(\mathbf{P1})~(B,\cdot_B,d_B)$ is a commutative (modified) $\lambda$-differential algebra and $f: (A,\cdot_A,d_A)\rightarrow (B,\cdot_B,d_B)$ is a (modified) $\lambda$-differential algebra homomorphism;
  \item $(\mathbf{P2})~g:(A,\{\cdot,\cdot\}_A,d_A)\rightarrow (B,\{\cdot,\cdot\}_B,d_B)$ is a $\lambda$-DL (or $\lambda$-mDL) algebra homomorphism, where $\{\cdot,\cdot\}_B$ is defined as the commutator: $\{x,y\}_B=x\cdot_B y-y\cdot_B x$, for $x,y \in B$;
  \item $(\mathbf{P3})~f(\{a,b\}_A)=g(a)\cdot_Bf(b)-f(b)\cdot_Bg(a)$;
  \item $(\mathbf{P4})~g(a\cdot_Ab)=f(a)\cdot_Bg(b)+f(b)\cdot_Bg(a)$,
\end{itemize}
then the triple $(B,f,g)$ is called a {\bf $\mathcal{P}$-triple} of $A$.

\end{defi}


\begin{rmk}\label{rmk:ex-triple}
The above properties, $(\mathbf{P1})\sim (\mathbf{P4})$ are equivalent to \eqref{eq:rela-M-pro} $\sim$ \eqref{eq:rela-H-pro}, respectively. Hence, the triple $(A^{ue},M,H)$ is a $\mathcal{P}$-triple of $A$.
\end{rmk}

\begin{thm}\label{thm:bi-commute}
Let $(B,f,g)$ be a $\mathcal{P}$-triple of $A$. Then there exists a unique (modified) $\lambda$-differential algebra homomorphism $\phi:A^{ue}\rightarrow B$ satisfying $f=\phi M$ and $g=\phi H$, i.e., the following diagram bi-commutes.
\[
 \xymatrix{
 A \ar[d]_{H} \ar[d]^{M} \ar[r]_{g} \ar[r]^{f} & B     \\
 A^{ue} \ar@{.>}[ur]|-{\exists!\phi}                               }
\]

\end{thm}

\begin{proof}
Assume that the homomorphism $\phi$ exists. For all $a\in A$, we obtain
\begin{equation*}
  \phi(M_a)=f(a),\quad\phi(H_a)=g(a).
\end{equation*}  
It suffices to prove that $\phi$ preserves all the relations \eqref{eq:rela-M-pro}$\sim$\eqref{eq:rela-M-unit} in Definition \ref{def:Uni-enve-alg}.

(P1) gives
\begin{equation*}
  \phi(M_{ab})=f(ab)=f(a)f(b)=\phi(M_{a})\phi(M_{b}),\quad\forall a,b\in A,
\end{equation*}
and $\phi(M_1)=f(1)=1=\phi(1)$, thus $\phi$ preserves \eqref{eq:rela-M-pro} and \eqref{eq:rela-M-unit}. 

For all $a,b\in A$, (P2) implies that
\begin{equation*}
  \phi(H_{\{a,b\}})=g(\{a,b\})=\{g(a),g(b)\}=g(a)g(b)-g(b)g(a)=\phi(H_a)\phi(H_b)-\phi(H_b)\phi(H_a),
\end{equation*}
hence $\phi$ preserves \eqref{eq:rela-H-bra}. Similarly, $\phi$ keeps \eqref{eq:rela-M-bra} and \eqref{eq:rela-H-pro} in turn because of (P3) and (P4).

Furthermore, for $a\in A$ we have
\begin{equation*}
  \phi D(M_a)=\phi(M_{d(a)})=f(d(a)),\quad\phi D(H_a)=\phi(H_{d(a)})=g(d(a)),
\end{equation*}
and it is clear that $\phi$ is unique. Consequently, the proof is finished.

\end{proof}

\begin{lem}\label{lem:AtensorB}
Let $(A,\cdot_A,\{\cdot,\cdot\}_A,d_A)$,~$(B,\cdot_B,\{\cdot,\cdot\}_B,d_B)$ be two $\lambda$-DP (or $\lambda$-mDP) algebras. If $(C,f,g)$ and $(D,j,k)$ are two $\mathcal{P}$-triples corresponding to $A$ and $B$, respectively. Then $T:=(C\otimes D,f\otimes j,f\otimes k+g\otimes j)$ is a $\mathcal{P}$-triple with respect to $A\otimes B$, where for all $a \in A, b \in B$,
\begin{equation*}
  (f\otimes j)(a\otimes b)=f(a)\otimes j(b).
\end{equation*}

\end{lem}

\begin{proof}
We give the proof only for (P1), (P4) in the differential case. 

(P1)~By Proposition \ref{pro:Poi-tensor-alg}, for all $a_1,a_2\in A, b_1,b_2\in B$, we have
\begin{eqnarray*}
  (f\otimes j)[(a_1\otimes b_1)(a_2\otimes b_2)] &=& (f\otimes j)(a_1a_2\otimes b_1b_2)=f(a_1a_2)\otimes j(b_1b_2) \\
  &=& f(a_1)f(a_2)\otimes j(b_1)j(b_2)=\big(f(a_1)\otimes j(b_1)\big)\big(f(a_2)\otimes j(b_2)\big) \\
  &=& (f\otimes j)(a_1\otimes b_1)(f\otimes j)(a_2\otimes b_2). 
\end{eqnarray*}
Furthermore, we have
\begin{eqnarray*}
  (f\otimes j)d_{A\otimes B}(a\otimes b) &=& (f\otimes j)\big(d_A(a)\otimes b+a\otimes d_B(b)+\lambda d_A(a)\otimes d_B(b)\big) \\
  &=& fd_A(a)\otimes j(b)+f(a)\otimes jd_B(b)+\lambda fd_A(a)\otimes jd_B(b)\\
  &=& d_Cf(a)\otimes j(b)+f(a)\otimes d_Dj(b)+\lambda  d_Cf(a)\otimes d_Dj(b)\\
  &=& d_{C\otimes D}(f\otimes j)(a\otimes b).
\end{eqnarray*}
Thus, $f\otimes j:A\otimes B\rightarrow C\otimes D$ is a $\lambda$-differential algebra homomorphism.

(P4)~On the one hand,
\begin{eqnarray*}
  (f\otimes k+g\otimes j)\big((a_1\otimes b_1)(a_2\otimes b_2)\big) &=& (f\otimes k+g\otimes j,a_1a_2\otimes b_1b_2)\\
  &=& f(a_1a_2)\otimes k(b_1b_2)+g(a_1a_2)\otimes j(b_1b_2)\\
  &=& f(a_1)f(a_2)\otimes\big(j(b_1)k(b_2)+j(b_2)k(b_1)\big)\\
  &+& \big(f(a_1)g(a_2)+f(a_2)g(a_1)\big)\otimes j(b_1)j(b_2)\\
  &=& f(a_1)f(a_2)\otimes j(b_1)k(b_2)+f(a_1)f(a_2)\otimes j(b_2)k(b_1)\\
  &+& f(a_1)g(a_2)\otimes j(b_1)j(b_2)+f(a_2)g(a_1)\otimes j(b_1)j(b_2).
\end{eqnarray*}
On the other hand,
\begin{eqnarray*}
  &&(f\otimes j)(a_1\otimes b_1)(f\otimes k+g\otimes j)(a_2\otimes b_2)+ (f\otimes j)(a_2\otimes b_2)(f\otimes k+g\otimes j)(a_1\otimes b_1) \\
  &=& \big(f(a_1)\otimes j(b_1)\big)\big(f(a_2)\otimes k(b_2)+g(a_2)\otimes j(b_2)\big)\\
  &+&\big(f(a_2)\otimes j(b_2)\big)\big(f(a_1)\otimes k(b_1)+g(a_1)\otimes j(b_1)\big)\\
  &=&f(a_1)f(a_2)\otimes j(b_1)k(b_2)+f(a_1)g(a_2)\otimes j(b_1)j(b_2)\\
  &+& f(a_2)f(a_1)\otimes j(b_2)k(b_1)+f(a_2)g(a_1)\otimes j(b_2)j(b_1).
\end{eqnarray*}
Note that
\begin{equation*}
  f(a_1)f(a_2)=f(a_2)f(a_1),\quad j(b_1)j(b_2)=j(b_2)j(b_1).
\end{equation*}
Replacing these things in the last equality, then (P4) holds. (P2) and (P3) can be proved similarly. 
The modified differential case can be proved similarly. This completes the proof. 

\end{proof}

\begin{thm}\label{thm:oppo-tensorpro}
Let $(A,\cdot_A,\{\cdot,\cdot\}_A,d_A)$,~$(B,\cdot_B,\{\cdot,\cdot\}_B,d_B)$ be two $\lambda$-DP (or $\lambda$-mDP) algebras. Then we obtain
\begin{description}
  \item[(i)] $(A^{op})^{ue}\cong(A^{ue})^{op}$.
  \item[(ii)] $(A\otimes B)^{ue}\cong A^{ue}\otimes B^{ue}$.
  \item[(iii)] $(A\otimes A^{op})^{ue}\cong A^{ue}\otimes (A^{ue})^{op}$.
\end{description}

\end{thm}

\begin{proof}
We only give the proof of (ii) here.
By Remark \ref{rmk:ex-triple}, we can obtain two $\mathcal{P}$-triples $(A^{ue}, M_A, H_A)$ and $(B^{ue}, M_B, H_B)$ corresponding to A and B, respectively. 

Suppose $(C, f, g)$ is a $\mathcal{P}$-triple of $A\otimes B$. On the one hand, by Theorem \ref{thm:bi-commute}, there exists a unique $\lambda$-DP (or $\lambda$-mDP) algebra homomorphism $\phi_{A\otimes B}:(A\otimes B)^{ue}\rightarrow C$.

On the other hand, by Lemma \ref{lem:AtensorB}, there is a $\mathcal{P}$-triple
\begin{equation*}
  T:=(A^{ue}\otimes B^{ue},M_A\otimes M_B,M_A\otimes H_B+H_A\otimes M_B)
\end{equation*}
of $A\otimes B$.

Next we show the following diagrams are commutative.
\[
\xymatrix{
& A\ar[dl]_-{i_A}\ar[dr]^-{f_A,g_A}\ar[rr]^-{M_A,H_A} & & A^{ue}\ar@{-->}[dl]^-{\exists !\phi_{A}}\ar[dr]^-{i_{A^{ue}}}&\\
A\otimes B\ar[rr]^-{f,g}&  &C & & A^{ue}\otimes B^{ue}\ar@{-->}[ll]_{\exists !\phi}\\
& B\ar[ur]_-{f_B,g_B}\ar[ul]^-{i_B}\ar[rr]_-{M_B,H_B}& & B^{ue}\ar@{-->}[ul]_-{\exists !\phi_{B}}\ar[ur]_-{i_{B^{ue}}}&\\
}
\]
For a natural inclusion $i_A:A\rightarrow A\otimes B$, we have $(C, f_A, g_A)$ is a $\mathcal{P}$-triple of $A$, where $f_A=fi_A$.
Therefore, by Theorem \ref{thm:bi-commute}, there exists a unique $\lambda$-DP (or $\lambda$-mDP) algebra homomorphism $\phi_A:A^{ue}\rightarrow C$ such that $f_A=\phi_AM_A$ and $g_A=\phi_AH_A$. Similarly, there exists a unique $\lambda$-DP (or $\lambda$-mDP) algebra homomorphism $\phi_B:B^{ue}\rightarrow C$ such that $f_B=\phi_BM_B$ and $g_B=\phi_BH_B$.       

Furthermore, we can obtain a unique $\lambda$-DP (or $\lambda$-mDP) algebra homomorphism $\phi:A^{ue}\otimes B^{ue}\rightarrow C$ making the above diagrams bi-commutative, where $i_{A^{ue}}:A^{ue}\rightarrow A^{ue}\otimes B^{ue}$ and $i_{B^{ue}}:B^{ue}\rightarrow A^{ue}\otimes B^{ue}$ are two natural inclusions.

Finally, by tracking the diagrams, $\phi$ is unique. Hence (ii) holds. We can get (i) by the similar way, and (iii) is a corollary of (i) and (ii).
\end{proof}

\end{document}